\documentclass[runningheads, envcountsame, a4paper]{llncs}

\usepackage{amsmath}
\usepackage{amssymb}

\usepackage{ltl}
\usepackage{wrapfig}
\usepackage{booktabs}
\usepackage{tikz}

\newcommand{\AP}{\textit{AP} }
\renewcommand{\next}{\LTLcircle}
\newcommand{\nexts}[1]{\next^{#1}}
\newcommand{\until}{{\cal U }}

\newcommand{\ceil}[1]{\lceil #1 \rceil}
\newcommand{\set}[1]{\{#1\}}
\newcommand{\nats}{\mathbb{N}}
\newcommand{\dotcup}{\ensuremath{\mathaccent\cdot\cup}}

\newcommand{\ltl}{\textsc{LTL}}

\newcommand{\class}{\mathcal{C}}

\newcommand{\pspace}{\textsc{Pspace}}

\newcommand{\exptime}{\textsc{Exptime}}

\newcommand{\expspace}{\textsc{Expspace}}

\newcommand{\twoexptime}{\textsc{2Exptime}}

\newcommand{\old}{\hspace{-0.04cm}_o}

\newcommand{\sharpp}{\#\textsc{P}}
\newcommand{\sharppspace}{\#\old\textsc{Pspace}}

\newcommand{\sharpexptime}{\#\old\textsc{Exptime}}
\newcommand{\sharpnexptime}{\#\old\textsc{NExptime}}
\newcommand{\sharpexpspace}{\#\old\textsc{Expspace}}

\newcommand{\sharptwoexptime}{\#\old\textsc{2Exptime}}

\newcommand{\sharppspaceprime}{\#\textsc{Pspace}}
\newcommand{\sharpexptimeprime}{\#\textsc{Exptime}}
\newcommand{\sharpexpspaceprime}{\#\textsc{Expspace}}
\newcommand{\sharptwoexptimeprime}{\#\textsc{2Exptime}}

\newcommand{\fpspace}{\textsc{FPspace}}

\begin{document}

\title{The Complexity of Counting Models of Linear-time Temporal Logic\thanks{This work was partially supported by the German Research Foundation (DFG) as
part of SFB/TR 14 AVACS and by the Deutsche Telekom Foundation.}
\fnmsep \thanks{A short version appears in the Proceedings of FSTTCS 2014.}
}
\author{Hazem Torfah \and Martin Zimmermann}
\institute{Reactive Systems Group,  Saarland University,  66123 Saarbr\"ucken, Germany \\
\email{$\{$torfah, zimmermann$\}$@react.uni-saarland.de}
}
\titlerunning{The Complexity of Counting Models of LTL}
\authorrunning{H. Torfah and M. Zimmermann}

\maketitle

\begin{abstract}
We determine the complexity of counting models of bounded size of specifications expressed in Linear-time Temporal Logic. 
Counting word-models is \#P-complete, if the bound is given in unary, and as hard as counting accepting runs of nondeterministic polynomial space Turing machines, if the bound is given in binary. Counting tree-models is as hard as counting accepting runs of nondeterministic exponential time Turing machines, if the bound is given in unary. For a binary encoding of the bound, the problem is at least as hard as counting accepting runs of nondeterministic exponential space Turing machines. On the other hand, it is not harder than counting accepting runs of nondeterministic doubly-exponential time Turing machines. 
\end{abstract}

\section{Introduction}

\emph{Model counting}, the problem of computing the \emph{number of models} of a logical formula, generalizes the satisfiability problem and has diverse applications: many probabilistic inference problems, such as Bayesian net reasoning \cite{Littman00stochasticboolean}, and planning problems, such as computing the robustness of plans in incomplete domains \cite{DBLP:conf/aaai/MorwoodB12}, can be formulated as model counting problems of propositional logic. Model counting for Linear-time Temporal Logic (LTL) has been recently introduced in \cite{DBLP:conf/lata/FinkbeinerT14}. LTL is the most commonly used specification logic for reactive systems \cite{Pnueli:1977:TLP:1382431.1382534} and the standard input language for model checking \cite{DBLP:series/faia/Biere09,Burch92symbolicmodel} and synthesis tools \cite{BGHKK12,DBLP:conf/cav/BohyBFJR12,DBLP:conf/tacas/Ehlers11}. LTL model counting asks for computing the number of transition systems that satisfy a given LTL formula. As such a formula has either zero or infinitely many models one considers models of \emph{bounded} size: for a formula $\varphi$ and a bound $k$, the problem is to count the number of models of $\varphi$ of size $k$. This is motivated by applications like bounded model checking \cite{DBLP:series/faia/Biere09} and bounded synthesis \cite{Finkbeiner+Schewe/2013/Bounded}, where one looks for short error paths and small implementations, respectively, by iteratively increasing a bound on the size of the model. Just like propositional model counting generalizes satisfiability, 
by considering two types of bounded models, namely, \emph{word-models} (of length~$k$) and \emph{tree-models} (of height~$k$), the authors of \cite{DBLP:conf/lata/FinkbeinerT14} introduced quantitative extensions of model checking and synthesis.

Word-models are ultimately periodic words of the form $u.v^\omega$ of bounded length $|u.v|$, which are used to model computations of a system. Counting word-models can be used in \emph{model checking} to determine not only the existence of computations that violate the specification, but also the \emph{number} of such \emph{violations}. To this end, one turns the model checking problem into an LTL satisfiability problem by encoding the transition system and the negation of the specification into a single LTL formula. Its models represent erroneous computations of the  system, i.e., counting them gives a quantitative notion of satisfaction. 

Tree-models are finite trees (of fixed out-degree)  of bounded height with back-edges at the leaves, i.e., tree-models can be exponentially-sized in the bound. They are used to describe implementations of the input-output behavior of reactive systems (see, e.g.,~\cite{Finkbeiner+Schewe/2013/Bounded}), namely the edges of a tree-model represent the input behavior of the environment and the nodes represent the corresponding output behavior of the system. In \emph{synthesis}, counting tree-models can be used to determine not only the existence of an implementation that satisfies the specification, but also the \emph{number} of such \emph{implementations}. This number is a helpful metric to understand how much room for implementation choices is left by a given specification, and to estimate the impact of new requirements on the remaining design space. 

For \emph{safety} LTL specifications~\cite{Sis94}, algorithms solving the word- and the tree-model counting problem were presented in \cite{DBLP:conf/lata/FinkbeinerT14}. The running time of the algorithms is linear in the bound and doubly-exponential respectively triply-exponential in the length of the formula. The high complexity in the formula is, however, not a major concern in practice, since specifications are typically small while models are large (cf.\ the state-space explosion problem).

Here, we complement these algorithms by analyzing the computational complexity of the model counting problems for \emph{full} LTL by placing the problems into counting complexity classes. 
These classes are based on counting accepting runs of nondeterministic Turing machines. In his seminal paper on the complexity of computing the permanent \cite{DBLP:journals/tcs/Valiant79}, Valiant introduced the class $\sharpp$ of counting problems associated with counting accepting runs of nondeterministic polynomial time Turing machines: a function $f\colon \Sigma^* \rightarrow \nats$ is in $\sharpp$ if there is a nondeterministic polynomial time Turing machine $\mathcal M$ such that $f(w)$ is equal to the number of accepting runs of $\mathcal M$ on $w$. Furthermore, for a class $\class$ of decision problems, he defined\footnote{Valiant originally used the notation~$\#\class$, but we added the subscript to distinguish the oracle-based classes from the classes introduced below.} $\#\old\class$ to be the class of counting problems induced by counting accepting runs of a nondeterministic polynomial time Turing machine with an oracle from $\class$. 

A nondeterministic polynomial time Turing machine $\mathcal M$ (with or without oracle) has at most $\mathcal{O}(2^{p(n)})$ different runs on inputs of length $n$ for some polynomial $p$. This means that there is an exponential upper bound on functions in $\sharpp$ and in $\#\old\class$ for every $\class$. However, an LTL tautology has exponentially many word-models of length $k$ and more than doubly-exponentially many tree-models of height $k$. This means, that no function in any of the counting classes defined above can capture the counting problems for LTL. 

To overcome this, we consider counting classes obtained by lifting the restriction on considering only nondeterministic polynomial time (oracle) machines: a function $f \colon \Sigma^* \rightarrow \nats$ is in $\sharppspaceprime$, if there is a nondeterministic polynomial \emph{space} Turing machine $\mathcal M$ such that $f(w)$ is equal to the number of accepting runs of $\mathcal M$ on $w$. The classes $\sharpexptimeprime$, $\sharpexpspaceprime$, and $\sharptwoexptimeprime$ are defined analogously\footnote{Following the ``satanic'' \cite{DBLP:journals/sigact/HemaspaandraV95} tradition of naming counting  classes, we drop the $\textsc{N}$ (standing for nondeterministic) in the names of the classes, just as it is done for $\sharpp$.}. Some of these classes appeared in the literature, e.g., $\sharppspaceprime$ was shown to be equal to $\fpspace$~\cite{Ladner89} (if the output is encoded in binary). Also,  computing a specific entry of a matrix power~$A^n$ is in $\sharppspaceprime$, if $A$ is represented succinctly and $n$ in binary~\cite{LohreyS14}, and counting self-avoiding walks in succinctly represented hypercubes is complete for $\sharpexptimeprime$~\cite{LiskiewiczOT03} under right-bit-shift reductions.

We place the LTL model counting problems in these classes. Unsurprisingly, the encoding of the bound $k$ is crucial: for unary bounds, we show counting word-models to be $\sharpp$-complete and show counting tree-models to be $\sharpexptimeprime$-complete. For binary bounds, the word-model counting problem is $\sharppspaceprime$-complete and counting tree-models is $\sharpexpspaceprime$-hard and in $\sharptwoexptimeprime$. The upper bounds hold for full LTL while the formulas for the lower bounds define safety properties (using only the temporal operators next and release). Thus, the lower bounds already hold for the fragment considered in~\cite{DBLP:conf/lata/FinkbeinerT14}.

The algorithms we present to prove the upper bounds are not practical since they are based on guessing a word (tree) and then model checking it. Hence, a deterministic variant of these algorithms would enumerate all words (trees) of length (height) $k$ and then run a model checking algorithm on them. In particular, the running time of the algorithms is exponential (or worse) in the bound $k$, which is in stark contrast to the practical algorithms~\cite{DBLP:conf/lata/FinkbeinerT14}. Our lower bounds are reductions from the problem of counting accepting runs of a Turing machine. For the word counting problem, the reductions are slight strengthenings of the reduction proving $\pspace$-hardness of the LTL model checking problem \cite{DBLP:journals/jacm/SistlaC85}. However, the reductions in the tree case are more involved (and to the best of our knowledge new), since we have to deal with exponential time respectively exponential space Turing machines. The main technical difficulties are to encode runs of exponential length and with configurations of exponential size into tree-models of ``small'' LTL formulas and to ensure that there is a one-to-one correspondence between accepting runs and models of the constructed formula.

All missing proofs can be found in the appendix.

\section{Preliminaries}

We represent models as \textit{labeled transition systems}. For a given finite set $\Upsilon$ of directions and a finite set $\Sigma$ of labels, a $\Sigma$-labeled $\Upsilon$-transition system is a tuple $\mathcal S = (S,s_0,\tau,o)$, consisting of a finite set of states~$S$, an initial state $s_0\in S$, a transition function $\tau\colon S\times \Upsilon \rightarrow S$, and a labeling function $o\colon S\rightarrow \Sigma$.
A \textit{path} in $\mathcal S$ is a sequence~$\pi\colon \mathbb N \rightarrow S \times \Upsilon$  of states and directions that follows the transition function, i.e., for all $i \in \mathbb N$  if $ \pi(i) = (s_i,e_i)$ and $\pi(i + 1) = (s_{i+1}, e_{i+1})$, then $s_{i+1} = \tau(s_i, e_i)$. We call the path initial if it starts with the initial state: $\pi(0) = (s_0 , e )$ for some $e \in \Upsilon$. 

We use Linear-time Temporal Logic (LTL) \cite{Pnueli:1977:TLP:1382431.1382534}, with the usual temporal operators Next $\LTLcircle$, Until $\cal U $, Release $\cal V$, and the derived operators Eventually $\LTLdiamond$ and Globally~$\LTLsquare$. We use $\nexts{i}$ to refer to  $i$ nested next operators. LTL formulas are defined over a set of atomic propositions $\AP = I \cup O$, which is partitioned into a set $I$ of input propositions and a set $O$ of output propositions. 
We denote the satisfaction of an LTL formula $\varphi$ by an infinite sequence $\sigma\colon \mathbb N \rightarrow 2^{\AP}$ of valuations of the atomic propositions  by $\sigma \models \varphi$. 
A $2^O$-labeled $2^I$-transition system $\mathcal S = (S,s_0,\tau,o)$ satisfies $\varphi$, if for every initial path $\pi$ the sequence $\sigma_{\pi}\colon i \mapsto o(\pi(i))$, where $o(s,e)= (o(s) \cup e)$, satisfies $\varphi$. Then $\mathcal S$ is a model of $\varphi$.

A \emph{$k$-word-model} of an LTL formula $\varphi$ over $\AP$ is a pair $(u,v)$ of finite words over $2^{\AP}$ such that $|u.v|=k$ and $u.v^\omega \models \varphi$. We call $u$ the prefix and $v$ the period of $(u,v)$. Note that an ultimately periodic word might be induced by more than one
$k$-word-model, i.e., $\{a\}^\omega$ is induced by the $2$-word-models $(\{a\}, \{a\})$ and $(\varepsilon, \{a\} \{a\})$.

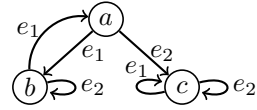
\begin{wrapfigure}{r}{0.27\textwidth}
\vspace{-1.7cm}
\center
\begin{tikzpicture}[inner sep=1pt,minimum size=0.45cm,node distance=2cm,semithick, scale = 1]
\node						at (0,.5) {};
\node[shape=circle,draw] (a1) at (0,0) { $a$ };
\node[shape=circle,draw] (a2) at (-1,-.9) { $b$ };
\node[shape=circle,draw] (a3) at (1,-.9) { $c$ };
\path[->,thick] (a1) edge node [ right=1mm]{$e_1$} (a2);
\path[->,thick] (a1) edge node [right]{$e_2$} (a3);
\path[->,thick] (a2) edge [bend left = 50] node [left]{$e_1$} (a1);
\path[->,thick] (a2) edge [loop right = 50] node [right]{$e_2$} (a2);
\path[->,thick] (a3) edge [loop right = 50] node [right]{$e_2$} (a3);
\path[->,thick] (a3) edge [loop left] node [above]{$e_1$} (a3);
\end{tikzpicture}
\vspace{-.1cm}
  \caption{A tree-model.}
 \label{fig:instances}
 \vspace{-.6cm}
 \end{wrapfigure}

A \emph{$k$-tree-model} of an LTL formula $\varphi$ over $\AP = I \cup O$ is a $2^O$-labeled~$2^I$-transition system that forms a tree (whose root is the initial state) of height $k$ with added back-edges from the leaves (for each leaf and direction, there is an edge to a state on the branch leading to the leaf) that satisfies $\varphi$. Figure~\ref{fig:instances} shows a tree-model of height one. 

Fix $\AP = I \cup O$. For a formula $\varphi$ and $k\in \mathbb{\nats}$, the \emph{$k$-word} (\emph{$k$-tree}) \emph{counting problem} asks to compute the number of $k$-word-models ($k$-tree-models up to isomorphism) of $\varphi$ over $\AP$.

\section{Counting Complexity Classes}
We use nondeterministic Turing machines with or without oracle access to define counting complexity classes, which we assume (without loss of generality) to terminate on every input. For background on (oracle) Turing machines and counting complexity we refer to~\cite{Arora:2009:CCM:1540612}.

A function~$f \colon \Sigma^* \rightarrow \nats$ is in the class~$\sharpp$~\cite{DBLP:journals/tcs/Valiant79} if there is a nondeterministic polynomial time Turing machine~$\mathcal M$ such that $f(w)$ is equal to the number of accepting runs of $\mathcal M$ on $w$. Similarly, for a given complexity class~$\class$ of decision problems, a function~$f$ is in $\#\old\class$~\cite{DBLP:journals/tcs/Valiant79,DBLP:journals/sigact/HemaspaandraV95} if there is a nondeterministic polynomial time oracle Turing machine~$\mathcal M$ with oracle in $\class$ such that $f(w)$ is equal to the number of accepting runs of $\mathcal M$ on $w$.
As a nondeterministic polynomial time Turing machine~$\mathcal M$ (with or without oracle) has at most $\mathcal{O}(2^{p(n)})$ runs on inputs of length~$n$ for some polynomial~$p$ (that only depends on~$\mathcal M$), we obtain an exponential upper bound on functions in $\sharpp$ and $\#\old\class$ for every $\class$, which explains the need for larger counting classes to characterize the model counting problems for $\ltl$.

A function~$f \colon \Sigma^* \rightarrow \nats$ is in $\sharppspaceprime$, if there is a nondeterministic polynomial \emph{space} Turing machine~$\mathcal M$ such that $f(w)$ is equal to the number of accepting runs of $\mathcal M$ on~$w$. $\sharpexptimeprime$, $\sharpexpspaceprime$, and $\sharptwoexptimeprime$  are defined by counting accepting runs of nondeterministic exponential time, exponential space, and  doubly-exponential time machines. 

\begin{proposition}\hfill
\begin{enumerate}

\item\label{oracleinclusions} 
$\sharpp \subseteq \sharppspace \subseteq \sharpexptime \subseteq \sharpnexptime \subseteq \sharpexpspace \subseteq$

 $\sharptwoexptime$.

\item\label{directinclusions}
$\sharppspaceprime \subseteq \sharpexptimeprime \subsetneq \sharpexpspaceprime \subseteq \sharptwoexptimeprime$.

\item $f \in \sharpexptimeprime$  implies $f(w) \in \mathcal{O}(2^{2^{p(|w|)}})$ for a polynomial~$p$.

\item $f \in  \sharptwoexptimeprime $ implies $f(w) \in \mathcal{O}(2^{2^{2^{p(|w|)}}})$ for a polynomial~$p$.

\item $w \mapsto 2^{2^{|w|}}$ is in $\sharppspaceprime $

\item $w \mapsto 2^{2^{2^{|w|}}}$ is in $\sharpexpspaceprime$. 

\end{enumerate}
\end{proposition}

We continue by relating the oracle-based and the generalized  classes introduced above.

\begin{lemma}\label{inclusions}
\label{oracle2directinclusions}
$\#\old{}C \subsetneq \#C$ for $C \in \set{\pspace, \exptime, \expspace, \twoexptime}$.
\end{lemma}

\begin{proof}
We show $\sharppspace \subsetneq \sharppspaceprime$, the other claims are proven analogously.
Let $f \in \sharppspace $, i.e., there is a nondeterministic polynomial time Turing machine~$\mathcal M$ with oracle~$A \in \pspace$ such that $f(w)$ is equal to the number of accepting runs of $\mathcal M$ on $w$. Note that all oracle queries are polynomially-sized in the length $|w|$ of the input to $\mathcal M$, since $\mathcal M$ is polynomially time-bounded. Hence, in nondeterministic polynomial space one can simulate $\mathcal M$ and evaluate the oracle calls explicitly by running a deterministic machine deciding $A$ in polynomial space. Since the oracle queries are evaluated deterministically, the simulation has as many accepting runs as $\mathcal M$ has. Thus, $f \in \sharppspaceprime$.

Now, consider the function~$|w| \mapsto 2^{2^{|w|}} $, which is in $\sharppspaceprime$, but not in $\sharppspace$.
\end{proof}

We use parsimonious reductions to define hardness and completeness, i.e., the most restrictive notion of reduction for counting problems. A counting problem~$f$ is $\sharpp$-hard, if for every $f' \in \sharpp$ there is a polynomial time computable function~$r$ such that $f'(x) = f(r(x))$ for all inputs~$x$. In particular, if $f'$ is induced by counting the accepting runs of $\mathcal M$, then $r$ depends on $\mathcal M$ (and possibly on its time-bound $p(n)$). Furthermore, $f$ is $\sharpp$-complete, if $f$ is $\sharpp$-hard and $f \in \sharpp$. Hardness and completeness for the other classes are defined analogously.

\section{Counting Word-Models}

In this section, we provide matching lower and upper bounds for the complexity of counting $k$-word-models of an LTL specification.

Our hardness proofs are based on constructing an LTL formula $\varphi_{\mathcal M}^w$ for a given Turing machine $\mathcal M$ and an 
input $w$ that encodes the accepting runs of $\mathcal M$ on $w$. Constructing such an LTL formula is straightforward and can be done in 
polynomial time for Turing machines with polynomially-sized configurations \cite{DBLP:journals/jacm/SistlaC85}. However, the challenge  
is to construct $\varphi_{\mathcal M}^w$ such that the number of accepting runs on $w$ is equal to the number of $k$-word-models of $\varphi_{\mathcal M}^w$ for a fixed  bound $k$.  
To this end, we have to enforce that each accepting run is represented by a unique $k$-word-model, i.e., by a unique prefix and period of total length $k$. We choose $k$ such that a run on $w$ of maximal length can be encoded in $k-1$ symbols and define $\varphi_{\mathcal M}^w$ such that it has only $k$-word-models whose period has length one. If a run of $\mathcal M$ is shorter than the maximal-length run we repeat the final configuration until reaching the maximal length,  which is achieved by accompanying the configurations in the encoding with consecutive id's.

For the upper bounds we show that there are appropriate nondeterministic Turing machines that guess an ultimately-periodic word and model check it against $\varphi$, i.e., the number of accepting runs on $k$ and $\varphi$ is equal to the number of $k$-word-models of $\varphi$.


\subsection{The Case of Unary Encodings.}

We show that counting word-models for unary bounds is $\sharpp$-complete.

\begin{theorem} The following problem is $\sharpp$-complete: Given an LTL formula~$\varphi$ and a bound~$k$ (in unary), how many $k$-word-models does $\varphi$ have?
\label{unaryComp}
\end{theorem}

\begin{proof}
We start with the hardness proof. Let $\mathcal M=(Q, q_\iota, Q_F,\Sigma, \delta)$ be a one-tape nondeterministic polynomial time Turing machine, where $Q$ is the set of states, $q_\iota$ is the initial state, $Q_F$ is the set of accepting states, $\Sigma$ is the alphabet, and $\delta\colon (Q \setminus Q_F)\times \Sigma \rightarrow 2^{Q\times \Sigma \times \{-1,1\}}$ is the transition function, where -1 and 1 encode the directions of the head. 
Note that the accepting states are terminal and that $\mathcal M$ rejects by terminating in a nonaccepting state. 
Let $\mathcal M$ be $p(n)$-time bounded for some polynomial $p$, and let $w=w_0 \cdots w_{n-1}$ be an input to $\mathcal M$. We construct an LTL formula $\varphi_{\mathcal M}^w$ and define a bound $k$, both polynomial in $|w|$ and $|\mathcal M|$, such that the number of accepting runs of $\mathcal M$ on $w$ is equal to the number of $k$-word-models of $\varphi_{\mathcal M}^w$.

 A run of $\mathcal M$ on $w$ is encoded by a finite alternating sequence of id's~$\text{id}_i$ and configurations~$c_i$ that is followed by an infinite repetition of a dummy symbol: 
\begin{align}	
	& \$\text{ id}_0~\#~c_0 ~\$\text{ id}_1 ~\#~ c_1~\$\text{ id}_2 ~\#~c_2~ \$ ~\cdots ~\$\text{ id}_{p(n)}~\#~c_{p(n)}~ (\bot)^\omega 
\label{wordshape}
\end{align}
Note that the period of the word-model is of the form~$\bot^\ell$ for some $\ell>0$. We will define $k$ such that maximal-length runs of $\mathcal M$ on $w$  can be encoded in the prefix, and such that the only possible period has length one by ensuring that exactly $p(n)$ configurations are encoded (by repeating the final configuration if necessary). This ensures that an accepting run is encoded by exactly one $k$-word-model.

 Let $l_r = p(n)$ be the maximal length of a run of $\mathcal M$ on $w$. The size of a configuration of $\mathcal M$ on $w$ is also bounded by $l_r$. For the id's we use an encoding of a binary counter with $l_c= \ceil{\text{log } l_r}$ many bits. Let $\AP= (Q \cup \Sigma) \dotcup \{b_1,\ldots,b_{l_c},\$,\#, \bot \}$ be the set of atomic propositions. The atomic propositions in $Q \cup \Sigma$ are used to encode the configuration of $\mathcal M$ by encoding the tape contents, the state of the machine, and the head position. The atomic propositions $b_1,\ldots,b_{l_c}$ represent the bit values of an id. The symbols \$ and  \#  are used as separators between id's and configurations, and $\bot$ is a dummy symbol for the model's period.  The distance between two \$ symbols and also between two \# symbols in the encoding is given by $d= l_r+3$ (see (\ref{wordshape})). Then, $\varphi_{\mathcal M}^w$ is the conjunction of the following formulas:
\begin{itemize}
	\item \emph{Id} encodes the id's of the configurations. It uses a formula~\emph{Inc}$(b_1, \ldots, b_{l_c}, d)$ that asserts that the number encoded by the bits~$b_j$ after $d$ steps is obtained by incrementing the number encoded at the current position. This formula will be reused in the tree case. 
	\item \emph{Init} asserts that the run of $\mathcal M$ starts with the initial configuration.
	\item \emph{Accept} asserts that the run must reach an accepting configuration.
	\item \emph{Config}  declares the consistency of two successive configurations with the transition relation of $\mathcal M$. Here, we use $d$ many next operators to relate the encoding of the two configurations.
	\item \emph{Repeat}  asserts that the encoding of an accepting configuration is repeated until the maximal id is reached
	\item \emph{Loop}  defines the period of the word-model, which may only contain~$\bot$.
\end{itemize}

All these properties can be expressed with polynomially-sized formulas, which can be found in the appendix. Furthermore, we need a formula to specify technical details: atomic propositions encoding the id's are not allowed to appear in the configurations and vice versa, symbols such as \$ and \#  only to appear as separators, each separator appears $p(n)$ times every $d$ positions, configuration encodings are represented by singleton sets of letters in $\Sigma$ with the exception of one set that contains a symbol from $Q$ to determine the head position and the state of $\mathcal M$, etc.

For $k= l_r\cdot (l_r+3) +1$, each accepting run of $\mathcal M$ on $w$ corresponds to exactly one $k$-word-model of $\varphi_{\mathcal M}^w$ that encodes the run in its prefix. Thus, the number of $k$-word-models is equal to the number of accepting runs of $\mathcal M$ on $w$. The formula~$\varphi_{\mathcal M}^w$ can be obtained in polynomial time in $|w|$ + $|\mathcal M|$, and $k$ (thus also its unary encoding) is polynomial in $|w|$.

To show that the problem is in $\sharpp$ we define a nondeterministic polynomial time Turing machine $\mathcal M$ as follows. $\mathcal M$ guesses a prefix $u$ and a period $v$ of an ultimately periodic word $u. v^\omega$ with $|u.v| = k$, and checks deterministically in polynomial time~\cite{KuhtzFinkbeiner09}, whether $u.v^\omega$ satisfies $\varphi$. Hence, for each $k$-word-model $(u,v)$ of $\varphi$ there is exactly one accepting run of $\mathcal M$. Thus, counting the $k$-word-models of $\varphi$ can be done by counting the accepting runs of $\mathcal M$ on the input $(k,\varphi)$.  
\end{proof}


\subsection{The Case of Binary Encodings.}
\label{subSec:wordBinary}
Now, we consider the word counting problem for binary bounds. As the input is  more compact, we have to deal with a larger complexity class. 

\begin{theorem} The following problem is $\sharppspaceprime$-complete: Given an LTL formula~$\varphi$ and a bound~$k$ (in binary), how many $k$-word-models does $\varphi$ have? 
\end{theorem}

\begin{proof}
The hardness proof is similar to the one for Theorem~\ref{unaryComp}: for a nondeterministic polynomial space Turing machine $\mathcal M$ bounded by a polynomial $p(n)$ and an input word $w$ we can define a formula $\varphi_{\mathcal M}^w$ in the same way as in Theorem~\ref{unaryComp}. The reason lies in that the size of configurations remains polynomial and the exponential number of configurations in a run can still be counted with a binary counter of polynomial size, i.e., we only have to use more bits~$b_j$ to encode the id's. Furthermore, we have to choose $k=2^{p'(n)}(p(n) +3)+1$ which can still be encoded using polynomially many bits. Here, $p'(n)$ is a polynomial (which only depends on $\mathcal M$) such that $\mathcal M$ terminates in at most $2^{p'(n)}$ steps on inputs of length~$n$.

For the proof of the upper bound we cannot just guess a $k$-model in polynomial space as in Theorem~\ref{unaryComp}, since the bound~$k$ is encoded in binary. Instead, we guess and verify the model on-the-fly relying on standard techniques for LTL model checking.

Formally, we construct a nondeterministic polynomial space Turing machine $\mathcal M$ which guesses a $k$-word-model $(u,v)$ by guessing $u\$v = w(0)\cdots w(i-1) \$$ $ w(i) \cdots w(k-1)$ symbol by symbol in a backwards fashion. Here, $\$$ is a fresh symbol to denote the beginning of the period. To meet the space requirement, $\mathcal M$ only stores the currently guessed symbol~$w(j)$, discards previously guessed symbols, and uses a binary counter to guess exactly $k$ symbols.

To verify whether $u.v^\omega$ satisfies $\varphi$, $\mathcal M$ also creates for every $j$ in the range~$0 \le j < k$ a set~$C_j$ of subformulas of $\varphi$ with the intention of $C_j$ containing exactly the subformulas which are satisfied in position~$j$ of $u.v^\omega$. Due to space-requirements, $\mathcal M$  only stores the set~$C_{k-1}$ as well as the sets $C_j$ and $C_{j+1}$, if $w(j)$ is the currently guessed symbol.
The set $C_{k-1}$ is guessed by $\mathcal M$ and the sets~$C_j$ for $j <k-1$ are uniquely determined by the following rules:

\begin{itemize}
	\item The membership of atomic propositions in $C_j$ is determined by $w(j)$, i.e., $C_j \cap \AP = w(j) $.
	\item  Conjunctions, disjunctions, and negations can be checked locally for consistency, e.g., $\neg \psi \in C_j $ if and only if $\psi \notin C_j$.
	\item $\next$-formulas are propagated backwards using the following equivalence: $\next\psi \in C_j $ if and only if $ \psi \in C_{j+1}$ (recall that $\mathcal M$ stores $C_j$ and $C_{j+1}$).
	\item $\until$-formulas are propagated backwards using the following equivalence: $\psi_0\until\psi_1 \in C_j $ if and only if $\psi_1 \in C_{j}$ or $\psi_0 \in C_j$ and $\psi_0 \until \psi_1 \in C_{j+1}$.
	\item $\cal V$-formulas can be rewritten into $\until$-formulas.
\end{itemize}
Once $\mathcal M$ has guessed the complete period~$v = w(i)\cdots w(k-1)$ it also checks that the guess of $C_{k-1}$ is correct (recall that $C_{k-1}$ is not discarded), which is the case if the following two requirements are met:
\begin{itemize}
	\item For every subformula~$\next\psi$ we have $\next\psi \in C_{k-1}$ if and only if $\psi \in C_i$.
	\item For every subformula~$\psi_0\until \psi_1$ we have $\psi_0\until\psi_1 \in C_{k-1} $ if and only if $\psi_1 \in C_{k-1}$ or $\psi_0 \in C_{k-1}$ and $\psi_0 \until \psi_1 \in C_{i}$. Furthermore, we have to require that $\psi_0 \until \psi_1 \in C_j$ for some $j$ in the range $i \le j < k$ implies $\psi_1 \in C_{j'}$ for some $j'$ in the range $i \le j' < k$. The latter condition can be checked on-the-fly while computing the $C_j$'s.
\end{itemize}
A straightforward structural induction over the construction of $\varphi$ shows that we have $\psi \in C_j$ if and only if $w(j)w(j+1)\cdots w(k-1)v^\omega \models \psi$ for every subformula~$\psi$ of $\varphi$. Hence, $u.v^\omega$ is a model of $\varphi$ if and only if $\varphi \in C_0$. Thus, $\mathcal M$ accepts if this is the case.
\end{proof}


\section{Counting Tree-Models}

In this section, we consider the tree counting problem for unary and binary bounds. There are at least doubly-exponentially many trees of height~$k$. Hence, if $k$ is encoded in binary, there are at least triply-exponentially many (in the size of the encoding of $k$) $k$-tree-models of a tautology. In order to capture these cardinalities using counting classes, we have to consider machines with that many runs, i.e., exponential time and exponential space machines.

In our hardness proofs, we again construct formulas $\varphi_\mathcal M^w$ that encode accepting runs of $\mathcal M$ on $w$ in trees. We choose binary trees, i.e., we consider a singleton set $I$ of input propositions. Recall that the power set of $I$ is used to (deterministically) label the edges in the tree. In the following, we identify the two elements of $2^I$ with the directions~$\texttt{left}$ and $\texttt{right}$. Note that we have to formalize the structure of our models and have to encode the runs of the machines using LTL. The semantics require a formula to be satisfied on all paths, which requires us to write conditional formulas of the form ``if the path has a certain form, then some property is satisfied''. We use two types of formulas: the ones of the first type describe the structure of the tree (e.g., it is complete and the targets of the back-edges) while the ones of the second type encode the actual run relying on this structure. The formulas of type one often assign addresses to nodes (sequences of bits that uniquely identify a leaf).

In the word case, we encoded runs of Turing machines whose configurations are of polynomial length. Hence, the distance between encodings of a tape cell in two successive configurations could be covered by a polynomial number of next-operators. Here, configurations are of exponential size. Thus, the challenge is to encode a run in a tree-model such that properties of two successive configurations can still be encoded by an LTL formula of polynomial size. We present two such encodings, one for unary and one for binary bounds.

For the upper bounds we show that there are appropriate nondeterministic machines that guess a finite tree with back-edges and model check it deterministically against $\varphi$, i.e., the number of accepting runs on $k$ and $\varphi$ is equal to the number of $k$-tree-models of $\varphi$.


\subsection{The Case of Unary Encodings} First, we consider tree-model counting for unary bounds.

\begin{theorem} The following problem is $\sharpexptimeprime$-complete: Given an LTL formula~$\varphi$ and a bound~$k$ (in unary), how many $k$-tree-models does $\varphi$ have?
\label{unaryCompTree}
\end{theorem}

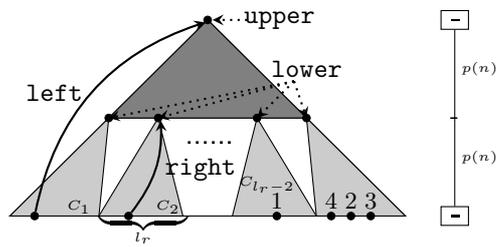
\begin{wrapfigure}{r}{0.53\textwidth}
\begin{center}
\vspace{-.8cm}
\begin{tikzpicture}[scale=.65,>=stealth]
\draw (0,0) -- (4,4) -- (8,0) -- cycle;
\draw (2,2) -- (6,2); 
\filldraw [fill= black!50] (2, 2) -- (4,4) -- (6, 2) -- cycle;
\filldraw [fill= black!20!white](0,0) -- (2,2) -- (1.8,0) -- cycle;
\node[circle,inner sep=1pt,draw,fill=black](r1)at(2,2){};
\filldraw [fill= black!20!white] (1.8, 0) -- (3,2) -- (3.5, 0) -- cycle;
\node[circle,inner sep=1pt,draw,fill=black] (r2) at (3,2){};
\filldraw [fill= black!20!white] (4.5, 0) -- (5,2) -- (6.2, 0) -- cycle; 
\node[circle,inner sep=1pt,draw,fill=black](b1)at(5,2){};
\filldraw [fill= black!20!white] (6.2, 0) -- (6,2) -- (8, 0) -- cycle;
\node[circle,inner sep=1pt,draw,fill=black](b1)at(6,2){}; 
\draw     [ dotted , line width=1](3.6,1.5) -- (3.8,1.5)-- (4.5,1.5);
\node[circle,inner sep=1pt,draw,fill=black](b0)at(4,4){};

\node [] (upper) at (5.45,4){\texttt{upper}}; \draw [->,  thick, dotted](4.8,4) -- (4.1,4);
\node [] (lower) at (6,3){\texttt{lower}}; \draw [->, thick,dotted](5.75,2.75) -- (2,2);\draw [->,thick, dotted](5.75,2.75) -- (3,2);\draw [->,thick, dotted](5.75,2.75) -- (5,2);\draw [->, thick,dotted](5.75,2.75) -- (6,2);

\node [] (h7) at (2.7,-.45) {\tiny $l_r$}; \node [] (brace) at (2.7,-.15) {$\underbrace{~~~~~~~~~~}$}; 
\node [] (h10) at (1.4,.5) [below] {\tiny$C_1$};
\node [] (h11) at (3.2,.5) [below] {\tiny$C_2$};
\node [] (h12) at (5.2,1) [below] {\tiny$C_{l_r-2}$};

\node [draw] (v1) at (9,0) {\_};
\node [draw] (v2) at (9,4) {\_}; 
\node [] (v3) at (9,2) {-};\node [] (v4) at (9.5,1.2) {\tiny$p(n)$}; \node [] (v5) at (9.5,3) {\tiny$p(n)$};

\path (v1) edge (v2) {};

\node [circle,inner sep=1pt,draw, fill=black] (l1) at (0.5,0) {};
\path [->,thick] (l1) edge [bend left= 30] node [left] {\texttt{left}} (3.95,3.95);
\node [circle,inner sep=1pt,draw, fill=black] (l2) at (2.4,0) {};
\path [->,thick] (l2) edge [bend right= 20] node [ right] {\texttt{right}} (3.05,1.95);

 \node [circle,inner sep=1pt,draw, fill=black] (l3) at (5.4,0) {};
 \node [] (alpha) at (5.4,0.3) {\small$1$};
\node [circle,inner sep=1pt,draw, fill=black] (l4) at (6.9,0) {};
\node [] (beta) at (6.9,0.3) {\small$2$};
\node [circle,inner sep=1pt,draw, fill=black] (l5) at (6.5,0) {};
\node [] (l) at (6.5,0.3) {\small$4$};
\node [circle,inner sep=1pt,draw, fill=black] (l6) at (7.3,0) {};
\node[] (r) at (7.3,0.3) {\small$3$};

\end{tikzpicture}
\vspace{-.7cm}
\end{center}

\caption{Encoding an exponentially long run in a tree-model of polynomial height. The configurations are encoded in the lower-trees (light gray subtrees).}
\label{blueredtree}
 \end{wrapfigure}

\noindent\textbf{Proof.} We start with the hardness proof. Let $\mathcal M=(Q, q_\iota, Q_F,\Sigma, \delta)$ be a one-tape nondeterministic exponential time Turing machine. Let $\mathcal M$ be $2^{p(n)}$-time bounded for a polynomial $p$ 
and let $w=w_0 \cdots w_{n-1}$ be an input to $\mathcal M$. 
We construct an LTL formula $\varphi_{\mathcal M}^w$ and define a bound $k$, both polynomial in $|w|$ and $|\mathcal M|$, 
such that the number of accepting runs of $\mathcal M$ on $w$ is equal to the number of $k$-tree-models of $\varphi_{\mathcal M}^w$.

A run of $\mathcal M$ is encoded in the leaves of a binary tree-model.
Let $l_r= 2^{ p(n)}$ be the maximal length of a run of $\mathcal M$ on $w$, which also bounds the size of a configuration. We choose $k=2p(n)$ to be the height of our tree-models. By using a formula labeling each of the first $k$ levels of the tree by a unique proposition we enforce that every model of height~$k$ is complete. Thus, it has $l_r^2$ many leaves, enough to encode a run of maximal length. Figure~\ref{blueredtree} shows the structure of our tree-model. 

Each configuration in the run is encoded in the leaves of a subtree of height $p(n)$, referred to as a \textit{lower}-tree (depicted by the light gray trees). The lower-trees are uniquely determined by a leaf of the \textit{upper}-tree (depicted in dark gray), which is the root of the lower-tree. By giving the leaves of the upper-tree id's, we also obtain unique id's for each of the lower-trees. These id's are used to enumerate the configurations of the run, i.e., two neighboring lower-trees encode two successive configurations of the run.  The id's can be determined by a binary counter with polynomially many bits. We also provide each leaf in a lower-tree with a unique id within this lower-tree. This is used to compare the contents of a tape cell in two successive configurations by comparing the labels of leaves with the same leaf id in two successive lower-trees. Thus, every leaf stores the id encoding of the configuration it is part of and the number of the cell it encodes.

Recall that in a tree-model each leaf has a back-edge for every direction. For the direction \texttt{left} we require a transition to the root of the upper-tree,  and for \texttt{right} a transition to the root of the own lower-tree. This enables us to compare two leaves in a lower-tree, or two leaves with the same id in two different lower-trees, with polynomially large formulas. 

The following formulas define the structure of our tree-models as explained above and also provide the nodes of the tree with  correct id's. We begin with \emph{Addr}(\texttt{root},$ a_1,\ldots,a_d$) which specifies a unique id for each leaf of a complete binary tree of height $d$ using bits $a_1,\ldots,a_d$, and provides the root of the tree with a label \texttt{root}. The id of a node depends on the sequence of \texttt{left} and \texttt{right} edges on the path from the root to this node, which is encoded in the bits~$a_1,\ldots,a_d$:
\begin{align*}
	\emph{Addr}(\texttt{root},a_1,\ldots,a_d)=  \texttt{root}  \wedge\hspace{-.1cm} \bigwedge \limits_{i=0}^{d-1}\left(\right. &\nexts{i}(\texttt{left} \rightarrow \nexts{d-i} \neg a_{i+1}) \wedge\\ &\nexts{i}(\texttt{right} \rightarrow \nexts{d-i}  a_{i+1})\left.\right)
\end{align*}
We use the formula \emph{Addr}(\texttt{upper},$u_1,\ldots,u_{p(n)}$) to address the upper-tree. This gives each lower-tree a unique id via the id of its root. We also supply each node in a lower-tree with the id of its root in the upper-tree:
 $ \bigwedge \nolimits_{p(n) \le i < k} \nexts{i}( \bigwedge \nolimits_{j=1}^{p(n)} (u_j \leftrightarrow \next u_{j})) $.
Furthermore, we use the formula~$\nexts{p(n)} \emph{Addr}(\texttt{lower},l_1,\ldots,l_{p(n)})$ to assign every leaf in a lower-tree a unique id within its lower-tree which essentially encodes the number of the tape cell it encodes. The next two formulas define the back-edges of the lower-trees. From each leaf, the \texttt{left} transition leads back to the root of the upper-tree  (recall that back-edges lead from a leaf to an ancestor), i.e., $ \nexts{{k}}(\texttt{left}\rightarrow  \next \texttt{upper})$,
and the \texttt{right} transition to the root of the lower-tree, i.e.,
        $ \nexts{{k}}(\texttt{right} \rightarrow  \next \texttt{lower} )$.
After setting up the structure of the trees, it remains to show how we encode a run in the leaves. We proceed with the same scheme as in the word case, and use the formula $\Delta_h(a_1,\ldots,a_m)$ which is satisfied, if and only if the bits $a_1,\ldots,a_m$ encode the number $h < 2^m $.
\begin{itemize}
	\item The formula \emph{Init} encodes the initial configuration in the lower-tree with id 0.
		\begin{align*}
				\nexts{k} \big[ \Delta_0(u_1,\ldots,u_{p(n)}) \rightarrow  & \big(
				( \Delta_0(l_1,\ldots,l_{p(n)}) \rightarrow q_\iota )\\
				& \wedge
				\bigwedge \limits_{0\le j < n} (\Delta_j(l_1,\ldots,l_{p(n)}) \rightarrow w_j )  \\
				&\wedge ((\bigwedge \limits_{0 \le j < n} \neg \Delta_j(l_1,\ldots,l_{p(n)})) \rightarrow \text{\textvisiblespace})~ \big)\big]
		\end{align*}
	\item The formula \emph{Accept} checks whether the rightmost lower-tree encodes an accepting configuration:
				$ \nexts{k} ((\Delta_{l_r}(u_1,\ldots,u_{p(n)}) \wedge \bigvee \nolimits_{q \in Q} q )\rightarrow \bigvee \nolimits_{q \in Q_F} q)$.
	\item The formulas $\emph{config}_{q,\alpha}$  and $\emph{config}_\alpha$ for states $q$ and symbols $\alpha$ encode the transition relation. For a leaf with labels $q$ and $\alpha $ (leaf $1$ in Figure~\ref{blueredtree}) and a transition $(q,\alpha, q',\beta, \texttt{dir})$, we have to check three leaves in the next lower-tree, namely, the leaf with the same id (leaf $2$) has to be labeled with $\beta$, and depending on \texttt{dir} either the successor leaf (leaf $3$) or the predecessor leaf (leaf $4$) has to be labeled with $q'$. The premise of the following formula only holds for paths that visit these leaves in the order given above, i.e., paths that lead to a leaf in a lower-tree, loop back to the root of the full tree and then lead to the same leaf id in the successor lower-tree (this takes $k+1$ edges), loop back to the root of this lower-tree and visit the leaf to the right (this takes $p(n)+1$ edges), back to the root of this lower-tree again and then to the leaf to the left (this takes $p(n)+1$ edges). To specify such a path, we use the formula \emph{Inc} to reach the successor leaf and a dual formula called \emph{Dec} to reach the predecessor leaf. This formula implements a decrement of a nonzero counter. Note that we have to require the paths to visit the successor and predecessor leaf in the next lower-tree, i.e., we have to check the bits~$u_j$ to reach the next lower-tree and the bits~$l_j$ to reach the leaves.
Thus, $\emph{config}_{q,\alpha}$ for $q \in Q\setminus Q_F$ is given by:
		\begin{align*}
				\nexts{k} \big[~ q & \wedge \alpha \wedge 	\emph{Inc}(u_1,\ldots,u_{p(n)},k+1) \wedge \bigwedge \nolimits_{i=1}^{p(n)} l_i \leftrightarrow \nexts{k+1} l_i)\\
				\wedge &\emph{Inc}(u_1,\ldots,u_{p(n)},k+p(n)+2)  \wedge \emph{Inc}(l_1,\ldots,l_{p(n)}, k+p(n)+2)\\
				 \wedge&\emph{Inc}(u_1,\ldots,u_{p(n)},k+2p(n)+3))  \wedge \emph{Dec}(l_1,\ldots,l_{p(n)},k+2p(n)+3)\\
				 \rightarrow &  \bigvee \nolimits_{{(q',\beta, \texttt{ dir}) \in \delta(q,\alpha)}} (\nexts{k+1} \beta \wedge \nexts{(k+1) + c_{\texttt{dir}} (p(n)+1) }q'  ) ~\big]
		\end{align*}	
Here, we have $c_{\texttt{dir}}=1$, if $\texttt{dir} = 1$, and $c_{\texttt{dir}}=2$, if $\texttt{dir} = -1$.
				
The formula $\emph{config}_\alpha$ determines the relation between the other tape cells' contents, namely where the head is not pointing to: 
		\begin{align*}
			\nexts{k} (\bigvee \limits_{i=1}^{p(n)} \neg u_i \wedge(\hspace{-2mm} \bigwedge \limits_{q \in Q \setminus Q_F} \hspace{-2mm} \neg q )\wedge \alpha  \wedge \emph{Inc}(u_1,\ldots,u_{p(n)},k+1)\\
			 \wedge  (\bigwedge \limits_{i=1}^{p(n)} l_i \leftrightarrow \nexts{k+1} l_i)  \rightarrow \nexts{k+1} \alpha)
		\end{align*}
	\item The formula \emph{Repeat} repeats accepting states in the next lower-tree, if the id of the current lower-tree is not maximal. The repetition of the letters is being taken care of by \emph{config}$_\alpha$.
		\begin{align*}
				  \nexts{k} \big[\big(\bigvee \limits_{i=1}^{p(n)} \neg u_i \wedge \emph{Inc}(u_1, \ldots,u_{p(n)},k+1) \wedge \bigwedge \limits_{i=1}^{p(n)} (l_i \leftrightarrow \nexts{k+1} l_i) \big)\rightarrow\\ \big(\hspace{-.2cm}\bigwedge \limits_{q_f \in Q_F}\hspace{-.2cm} q_f \rightarrow \nexts{k+1} q_f \big)\big]
		\end{align*}
\end{itemize}

Similar to the word case we need some additional formulas to prevent atomic propositions of configurations to appear elsewhere in the tree to guarantee the one-to-one relation between runs and tree-models. For example to prevent a state label from appearing twice in a configuration we use a formula that asserts that from a leaf in which a state is encoded, no other leaf with a state label is reachable within $p(n)+1$ steps, i.e., in the same lower-tree. This ensures that every configuration has exactly one state.

To show that the problem is in $\sharpexptimeprime$ we define a nondeterministic exponential time Turing machine $\mathcal M$ as follows.
$\mathcal M$ guesses a tree of height~$k$ (which is of exponential size) and checks whether it satisfies $\varphi$ using the classical model checking algorithm: $\mathcal M$ constructs the B\"uchi automaton recognizing the language of $\neg \varphi$ and checks whether the product of the tree and the automaton has an empty language. The automaton and the product are of exponential size and the emptiness check can be performed in deterministic polynomial time (in the size of the product). Hence, $\mathcal M$ runs in exponential time in $k$ and the size of $\varphi$. 
For each $k$-tree-model of $\varphi$, there is exactly one accepting run in $\mathcal M$. Thus, counting the $k$-tree-models of $\varphi$ can be done by counting the accepting runs of $\mathcal M$ on the input  $(k,\varphi)$.  \qed


\subsection{The Case of Binary Encodings.}
\label{subSec:treeBinary}
In this section, we consider tree-model counting  for binary bounds. Since the bound is encoded compactly, the trees we work with have exponential height and therefore doubly-exponential size. Unfortunately, our upper and lower bounds do not match (see the discussion in the next section). 

\begin{theorem} The following problem is $\sharpexpspaceprime$-hard and in $\sharptwoexptimeprime$: Given an LTL formula~$\varphi$ and a bound~$k$ (in binary), how many $k$-tree-models does $\varphi$ have?
\label{binaryHardnessTree}

\end{theorem}
\begin{proof}
Let $\mathcal M=(Q, q_\iota, Q_F,\Sigma, \delta)$ be a one-tape nondeterministic exponential space Turing machine
and let $w=w_0 \cdots w_{n-1}$ be an input to $\mathcal M$. Furthermore, let $l_c=2^{p(n)}-2$ be the maximal configuration length (for some polynomial~$p$) and
let $l_r= 2^{2^{ p'(n)}}$ be the maximal length of a run of $\mathcal M$ on $w$ ($p'$ is a polynomial which only depends on $\mathcal M$).

 We choose $k=m \cdot 2^{p'(n)}$ to be the height of our tree-models, where $m$ is the smallest power of two greater than $p(n)$.
Figure \ref{fig:DFS} shows the main structure of our tree-models. We use nonbalanced binary trees that are composed of trees of height $m$. We refer to the latter trees as the \emph{inner}-trees. The  outermost leaves of an inner-tree are inner nodes and the others are leaves in the tree-model. Hence, each inner-tree has two children, which are again inner-trees rooted at the leftmost respectively the rightmost leaf.

In each inner-tree, we will encode a configuration in a similar way as in the unary case (Theorem \ref{unaryCompTree}), namely in the leaves (except the  two leaves serving as roots for other inner trees, which explains the $-2$ in the definition of $l_c$). We encode the configurations of a run in the tree-model such that we traverse the inner-trees in a depth-first search manner (DFS). In Figure \ref{fig:DFS}, we can see how a run of 16 configurations can be encoded in a tree-model with four layers of inner-trees. To encode the DFS structure, we label each root of an inner-tree with its \emph{level} (the number of inner-tree ancestors) and with its so-called \emph{right-child-depth}: the number of right-child-inner-trees visited since the last left child to reach this tree (e.g., this value is $0$ for the left children $C_1,C_2,C_3,C_7$; it is $1$ for $C_6$ and $3$ for $C_{15}$). This will help to determine the next inner-tree in line in the DFS structure. We need a polynomial number of bits to encode these addresses. With the \texttt{right} transition we allow the leaves of an inner-tree to reach its root and we use \texttt{left} in the inner-tree of maximal level to reach the parent of the next inner-tree in DFS order. In this way, the distance between the encoding of a tape cell in two successive configurations is polynomial.

\begin{figure}[t]
	\center
	\begin{tikzpicture}[>=stealth]
		
		\draw (0,0) -- (.5,.5) -- (1,0) -- cycle;
		\draw (1.1,0) -- (1.6,.5) -- (2.1,0) -- cycle;
		\draw (2.3,0) -- (2.8,.5) -- (3.3,0) -- cycle;
		\draw (3.4,0) -- (3.9,.5) -- (4.4,0) -- cycle;
		\draw (4.6,0) -- (5.1,.5) -- (5.6,0) -- cycle;
		\draw (5.7,0) -- (6.2,.5) -- (6.7,0) -- cycle;
		\draw (6.9,0) -- (7.4,.5) -- (7.9,0) -- cycle;
		\draw (8,0) -- (8.5,.5) -- (9,0) -- cycle;
		
		\draw (.5,.5) -- (1,1) -- (1.6,.5) -- cycle;
		\draw (2.8,.5) -- (3.3,1) -- (3.9,.5) -- cycle;
		\draw (5.1,.5) -- (5.6,1) -- (6.2,.5) -- cycle;
		\draw (7.4,.5) -- (7.9,1) -- (8.5,.5) -- cycle;
		
		\draw (1,1) -- (2.2,1.5) -- (3.3,1) -- cycle;
		\draw (5.6,1) -- (6.7,1.5) -- (7.9,1) -- cycle;
		
		\draw (2.2,1.5) -- (4.4,2) -- (6.7,1.5) -- cycle;
		
		\node[](t1)at(4.4,1.8){\tiny $C_1$};
		\node[](t2)at(2.2,1.3){\tiny $C_2$};
		\node[](t3)at(1.05,.7){\tiny $C_3$};
		\node[](t4)at(.55,.2){\tiny $C_4$};
		\node[](t5)at(1.65,.2){\tiny $C_5$};
		\node[](t6)at(3.35,.7){\tiny $C_6$};
		\node[](t7)at(2.85,.2){\tiny $C_7$};
		\node[](t8)at(3.95,.2){\tiny $C_8$};
		\node[](t9)at(6.7,1.3){\tiny $C_9$};
		\node[](t10)at(5.65,.7){\tiny $C_{10}$};	
		\node[](t11)at(5.15,.2){\tiny $C_{11}$};
		\node[](t12)at(6.25,.2){\tiny $C_{12}$};
		\node[](t13)at(7.95,.7){\tiny $C_{13}$};
		\node[](t14)at(7.45,.2){\tiny $C_{14}$};
		\node[](t15)at(8.55,.2){\tiny $C_{15}$};
		
				\path [->](0,0) edge [dashed,thick, bend left=60] node [left]{\texttt{left}} (1,1);
				\path [->](2.1,0) edge [dashed,thick, bend right=30] node [left]{} (2.2,1.5);
				\path [->](3.3,0) edge [dashed,thick, bend left=20] node [left]{} (3.3,1);
				\path [->](4.4,0) edge [dashed,thick, bend right=40] node [left]{} (4.4,2);
				\path [->](4.6,0) edge [dashed,thick, bend left=30] node [left]{} (5.6,1);
				\path [->](6.7,0) edge [dashed,thick, bend right=30] node [left]{} (6.7,1.5);
				\path [->](6.9,0) edge [dashed,thick, bend left=30] node [left]{} (7.9,1);
				\path [->](9,0) edge [bend right = 30, thick,dashed] node [left]{} (8.5,.5);

				\node[circle,inner sep=1pt,draw,fill=black](r0)at(3,1.5){};
				\path [->](r0) edge [thick, bend left=30] node [left]{\texttt{right}} (4.4,2);
				
		\node [draw] (v1) at (10.5,0) {\_};
		\node [draw] (v2) at (10.5,2) {\_};
		\node [] (v3) at (10.7,.5) {~- $_{3m}$}; 
		\node [] (v4) at (10.7,1) {~- $_{2m}$};
		\node [] (v5) at (10.7,1.5) {- $_{m}$};
		\path (v1) edge (v2) {};
			
	\end{tikzpicture}
\caption{Tree-model with  DFS structure.}
\label{fig:DFS}

\end{figure}
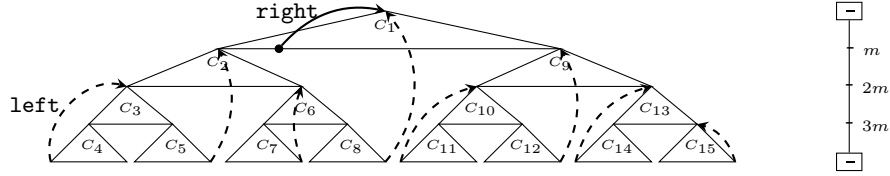
As the distance between an inner-tree and its successor is polynomial, the formulas for encoding the run in the tree-model adapt the ideas of the formulas in the unary case with slight modifications that deal with the DFS order of inner-trees. A detailed description of the construction can be found in the appendix.

The  upper bound is proved using the same algorithm as in the proof of Theorem~\ref{unaryCompTree}.
\end{proof}

\section{Discussion}

We investigated the complexity of the model counting problem for specifications in Linear-time temporal logic. The word-model counting problems are $\sharpp$-complete (for unary bounds) respectively $\sharppspaceprime$-complete (for binary bounds) while the tree-model counting problems are $\sharpexptimeprime$-complete respectively $\sharpexpspaceprime$-hard and in $\sharptwoexptimeprime$, i.e., the exact complexity of the tree-model counting problem for binary bounds is open.

The problem we face trying to lower the upper bound is that we cannot guess the complete tree-model in nondeterministic exponential space. To meet the space-requirements, we have to construct it step by step, as in the proof of the corresponding upper bound in the word case. However, the correctness of the on-the-fly model checking procedure described there relies on the fact that the model is an ultimately-periodic word. It is open whether the technique can be extended to tree-models. On the other hand, if we try to raise the lower bound, we have to encode doubly-exponential time Turing machines, which seems challenging using polynomially-sized LTL formulas. 

To conclude, let us mention another variation of the model counting problem: counting arbitrary transition systems, where the bound~$k$ now refers to the size of the transition system. 
For unary bounds, the problem is $\sharpp$-hard, which can be shown by strengthening Theorem~\ref{unaryComp}, and in $\sharppspace$, since LTL model checking is in $\pspace$. For binary bounds, the construction presented in Theorem~\ref{unaryCompTree} yields $\sharpexptimeprime$-hardness and the problem is in $\sharpexptimeprime$, which can be shown by adapting the algorithm presented in the theorem.

\subsection*{Acknowledgments.} We would like to thank Markus Lohrey and an anonymous reviewer for bringing Ladner's work on polynomial space counting~\cite{Ladner89} to our attention.


\bibliographystyle{plain}
\bibliography{biblio}


\newpage
\appendix

\section{Appendix}

\subsection{Proof of Theorem~\ref{unaryComp}}
\label{app:unarySPHard}
In this section, we present the formulas omitted in the proof of $\sharpp$-hardness of the word counting problem for unary bounds.

We start with the formula \emph{Id}, which uses the formula \emph{Inc} that enforces an increment of a binary counter. For later reuse, \emph{Inc} is parameterized by the propositions $b_1,\ldots, b_\ell$ encoding the bits and the distance~$d$ between the two positions to be compared. Intuitively, the different subformulas distinguish whether the increment ripples through to the current bit $b_i$ or not. Note that the increment property only has to hold if there is no overflow of the counter.
\begin{align*}
 \emph{Inc}(b_1,\ldots,b_{\ell},d) = ( \bigvee\limits_{0 < i \le \ell} \neg b_i ) \rightarrow \bigwedge \limits_{{0<i\leq  \ell }} \big[\,
& ((\neg b_i \wedge \phantom{\neg}\bigwedge \limits_{{i < j \le \ell }} \phantom{\neg} b_j) \rightarrow \nexts{{d}} \phantom{\neg} b_i)\\
\wedge&  ((\neg b_i \wedge {\neg}\bigwedge \limits_{{i < j \le \ell }} \phantom{\neg} b_j) \rightarrow \nexts{{d}} {\neg} b_i)\\
\wedge&  ((\phantom{\neg} b_i \wedge \phantom{\neg}\bigwedge \limits_{{i < j \le \ell }} \phantom{\neg} b_j) \rightarrow \nexts{{d}} {\neg} b_i)\\
\wedge&  ((\phantom{\neg} b_i \wedge {\neg}\bigwedge \limits_{{i < j \le \ell }} \phantom{\neg} b_j) \rightarrow \nexts{{d}} \phantom{\neg} b_i)
 \,\big]\\ 
\end{align*}
Now, the formula~\emph{Id} is defined by initializing the counter to zero and always requiring an increment after a $\$$-separator:
\begin{align*}
	\emph{Id} = \$ ~ \wedge ~ \next( \bigwedge\limits_{0<j\le l_c}\neg b_j	) ~ \wedge ~ \LTLsquare (\$ \rightarrow  \next \emph{Inc}(b_1, \ldots, b_{l_c}, d)) 
\end{align*}

We continue with the formula \emph{Init}. In the initial configuration the tape of $\mathcal M$ contains the input word $w$,  the head is on the first cell, and $\mathcal M$ is in its initial state:
 \begin{align*}
 	\emph{Init}=  \nexts{{2}} (\#~ \wedge ~\next q_{\iota} ~\wedge~(\bigwedge \limits_{0 \leq j < n} \nexts{{j}} w_j) ~ \wedge~ (\bigwedge \limits_{n \le j \leq l_r} \nexts{{j}} \text{\textvisiblespace})~)
 \end{align*}
The symbol \text{\textvisiblespace} ~refers to the blank cells of the tape.

 The formula \emph{Accept} considers the maximal id and checks whether it is followed by an accepting configuration:
\begin{align*}
	\emph{Accept}=\LTLsquare (\$ \wedge \next(\bigwedge \limits_{0<j\leq l_c} b_j )~~ \rightarrow ~~ \bigvee \limits_{q\in Q_F} \bigvee \limits_{0<j\leq l_r} \nexts{{j+2}}q)
\end{align*}
 For  atomic propositions $q\in Q\setminus Q_F$ and $\alpha \in\Sigma$ a formula $\textit{config}_{q,\alpha}$ asserts the relation between the states, the head positions, and the content of the cell where the head is pointing to in two successive configurations:
\begin{align*}
 \textit{config}_{q,\alpha} = \LTLsquare(~q \wedge \alpha \rightarrow  \bigvee  \limits_{_{(q',\beta, \text{dir}) \in \delta(q,\alpha)}} \nexts{d} \beta \wedge \nexts{{d+\text{dir}}}q'~~ )
\end{align*}
Another formula $\textit{config}_\alpha$ asserts the relation of the other tape cells of successive configurations; the content of these cells is copied, unless the id is maximal\footnote{Note that this is not necessary for $\textit{config}_{q,\alpha}$ since the machine terminates in at most $p(n)$ steps}:
\begin{align*}
  \textit{config}_{\alpha} = \LTLsquare( \$ \wedge \next (\bigvee \limits_{0<j\le l_c} \neg b_j) \rightarrow    \bigwedge \limits_{0<j\le l_r} \nexts{j+2}( (\bigwedge 
  \limits_{q\in Q\setminus Q_F} \neg q ) \wedge \alpha \rightarrow  \nexts{d} 
  \alpha ))
\end{align*}
\emph{Config} is the conjunction of all formulas $\textit{config}_{\alpha} $ and $\textit{config}_{q,\alpha}$. 

The formula \emph{Repeat} requires an accepting configuration to be repeated if the id is not yet maximal. The repetition of the letters is taken care of by the formulas $\textit{config}_{\alpha} $. Hence, \emph{Repeat} only requires to copy the state and the head position.

\begin{align*}
\emph{Repeat}=\LTLsquare (\$ \wedge (\next \bigvee \limits_{_{0<j\leq { l_c}}}\neg b_j) \rightarrow
 \bigwedge \limits_{q_f \in Q_F	}  \bigwedge \limits_{0<j\leq l_r} \next^{^{j+2}}(q_f  \rightarrow \nexts{{d}}  q_f))
\end{align*}
Finally, the period of the model is fixed by the formula \emph{Loop} which requires the symbol $\bot$ to be repeated after reaching the configuration with the maximal id:

\begin{align*}
 \emph{Loop}=\LTLsquare (\$ \wedge \next(\bigwedge \limits_{0<j\leq {l_c}}  b_j) \rightarrow \nexts{{l_r+2}}~ \LTLsquare \bot)
\end{align*}

\subsection{Proof of Theorem~\ref{binaryHardnessTree}}

The following formulas define the structure of our tree models and also provide them with the correct level and right-child-depth.  
	We use propositions $\iota_1,\ldots,\iota_m$ to give \emph{i}ds to the leaves of 
	an inner-tree. For the $2^{p'(n)}$ different  levels of  inner-trees in our tree models we use propositions 
	$l_1,\ldots,l_h$, where $h =p'(n)$, to encode for each inner-tree its  \emph{l}evel. We also give internally for each node in an inner-tree its \emph{d}epth in this tree via $d_1,\ldots,d_{{\text{log}(m)}}$ (remember that $m$ is a power of $2$). Propositions $r_1,\ldots,r_h$ are used to determine the inner-tree's \emph{r}ight-child-depth. The maximum right-child-depth that can be reached is $2^{p'(n)}$, namely for the rightmost inner-tree at the maximal level.
	
We start by labeling each root of an inner-tree (and no other vertices) with the label \texttt{root}:
\begin{align*}
		&\texttt{root} \wedge \\
		&\LTLsquare \big[\,\big( \texttt{root} \wedge \neg (\bigwedge \limits_{0< i \leq h} l_i ) \wedge ((\bigwedge \limits_{0\le j < m} \nexts{j} \texttt{left})\vee (\bigwedge \limits_{0\le j <m} \nexts{j} \texttt{right}))\big)\leftrightarrow \nexts{m} \texttt{root} )\,\big]
\end{align*}
The negation of the big conjunction over propositions $l_i$ is used to exclude inner-trees of the last level. How the levels are defined in the tree is shown in more detail in the formula \emph{depth}.

We encode a configuration in the leaves of an inner-tree in the same way as in the unary case (Lemma \ref{unaryCompTree}). Therefore, we provide the leaves with unique id's, which enable us to compare the contents of tape cells in two successive configurations by comparing leaves with the same id in two successive inner-trees. To this end, we use the formula  \textit{Addr}, as defined in the unary case, to equip the leaves with unique id's within their inner-tree:
\begin{align*}
	\LTLsquare(\texttt{root} \rightarrow \textit{Addr}(\texttt{root},\iota_1,\ldots,\iota_{m}))
\end{align*}

The next formula uses the propositions $l_1,\ldots,l_h$ to supply each inner-tree with its level, which is equal to the number of \texttt{root} labels visited from the root to this inner-tree. Line~(1) assigns the root of the tree model with level 0. Line~(2) assigns each node in an inner-tree, with the exception of the leaves labeled with \texttt{root} (the  outermost leaves in all levels but the last), with the level of its inner-tree root. Line~(3) gives, inductively, each root of an inner-tree its inner-tree level. 
\setcounter{equation}{0}
\begin{align}
		\emph{depth}= &\bigwedge \limits_{0<j\le h} \neg l_j\\
		&\wedge \LTLsquare (\texttt{root} \rightarrow \bigwedge \limits_{0 \le j < m} \nexts{j}(\next \neg \texttt{root} \rightarrow \bigwedge \limits_{0<j\le h} ( l_j \leftrightarrow \next l_j ))\\
					  &  \wedge \LTLsquare (\texttt{root} \wedge (\bigvee \limits_{0<j\le h} \neg l_j) \wedge \nexts{m} \texttt{root} \rightarrow  \emph{Inc}(l_1,\ldots,l_{h},m)
\end{align}

If \emph{depth} is satisfied then the tree contains a doubly-exponential number of inner-trees, enough to encode a run of $\mathcal M$ on $w$.

The following formula  gives each inner-tree its right-child-depth in the tree model. If we are at a root of an inner-tree we reach the root of its left and right child in $m$ steps via the  outermost leaves. The formula counts the number of right-children visited along the way including the visited tree (Line~(3)). Every time we visit a left-child the counter is reset (Line~(4)). We again supply each node in an inner-tree, with the exception of the outermost leaves (Line (5)), with the right-child-depth of this tree. However the outermost leaves of all maximal-level inner-trees  are labeled with the right-child-depth of the inner-tree.

\setcounter{equation}{0}
\begin{align}
	\emph{rLevel}= &\bigwedge \limits_{0<j\le h}  \neg r_j &  \\
				    & \wedge \LTLsquare (\texttt{root} \wedge \bigvee \limits_{0<j\le h} \neg l_j  \rightarrow & \\
					& \quad\quad \quad \quad(( \bigwedge \limits_{0<j \le m} \nexts{j}\texttt{right}) \rightarrow \emph{Inc}(r_1,\ldots,r_{h},m))\\
					& \quad\quad\quad \wedge ((\bigwedge \limits_{0<j \le m} \nexts{j}\texttt{left} )\rightarrow \bigwedge \limits_{0<j\le h}\neg r_j))\\
					& \wedge \LTLsquare ( \texttt{root} \rightarrow \bigwedge \limits_{0\le j <m} \nexts{j}(\next\neg \texttt{root} \rightarrow \bigwedge \limits_{0<j\le h} r_j \leftrightarrow \next r_j))  )
 \end{align}

Now that we have defined all the id's we need we show how to route the transitions at the leaves to the DFS positions as described earlier. To compute the correct node we only need to compute its level, because in the definition of our tree models we force a back-edge to go to an ancestor node. This level can be computed as follows. If we are at a leaf of an inner-tree in the maximal level, reached by $j$ many right children since the last left child the back-edge goes to the root with level $2^{p'(n)} -(j+1)$ (cf.\ Figure \ref{fig:DFS}). This can be formulated by incrementing the right-child-depth  described with bits $r_1,\ldots,r_h$ and inverting the bits of the result. The propositions $d_1,\ldots,d_{{\text{log}(m)}}$ are used to talk about the leaves of an inner-tree. If we are at a leaf in an inner-tree of maximal level we move with direction \texttt{left} to the next DFS position, namely the root of the next inner-tree in the DFS Line~(computed by \emph{addNflip}):

\begin{align*}
	&\emph{DFS}=\\ &\LTLsquare ( \bigwedge \limits_{0<j\le {\text{log}(m)}} d_j \wedge \bigwedge \limits_{0<j\le h} l_j \wedge \texttt{left} \rightarrow \next \texttt{root} \wedge \emph{addNflip}(l_1,\ldots,l_{h},r_1,\ldots,r_{h}))
\end{align*}
where:

\begin{align*}
\emph{addNflip}(l_1,\ldots,l_{c},r_1,\ldots ,r_{c})=~    \bigwedge \limits_{{0<i\leq  c}}\big[
& ((\neg r_i \wedge \phantom{\neg}\bigwedge \limits_{{i < j \leq  c}} r_j) \rightarrow \next \neg l_i)\\
\wedge &  ((\neg r_i \wedge \neg \bigwedge \limits_{{i < j \leq  c }}  r_j) \rightarrow \next  \phantom{\neg} l_i)\\
\wedge &  (( \phantom{\neg} r_i \wedge \phantom{\neg}\bigwedge \limits_{{i < j \leq  c}}   r_j) \rightarrow \next \phantom{\neg} l_i)\\
\wedge &  (( \phantom{\neg} r_i \wedge \neg \bigwedge \limits_{{i < j \leq  c}} r_j) \rightarrow \next \neg l_i)
\end{align*}
Notice that \emph{addNflip} is similar to \emph{Inc} with the difference of flipping the bit.

Finally, a formula that asserts that from leaves of maximal  level inner-trees a \texttt{right} transition goes to the root of these trees (Line~(1)). For leaves of inner-trees of non-maximal level (the outermost leaves are not considered leaves of the tree model as their \texttt{right} and \texttt{left} transitions lead to subtrees), we move with both \texttt{left} and \texttt{right} to the root of their inner-trees (Line~(2)):  
\setcounter{equation}{0}
\begin{align}
	&~~~\LTLsquare (\bigwedge \limits_{0<j \le { \text{log}(m)}} d_j  \wedge  \bigwedge \limits_{0<j \le h} l_j   \wedge \texttt{right} \rightarrow \next (\texttt{root} \wedge \bigwedge \limits_{0<j \le h} l_j))\\
	& \wedge \LTLsquare (\bigwedge \limits_{0<j \le { \text{log}(m)}} d_j  \wedge ( \bigvee \limits_{0<j \le h} \neg l_j ) \wedge \neg \texttt{root} \rightarrow \next \texttt{root} \wedge  (\bigwedge \limits_{0<j\le h} l_j \leftrightarrow \next l_j)) 
\end{align}
Now, we present the formulas that encode the run of the Turing machine. Here, we again use the formula~$\Delta_h(a_1, \ldots, a_m)$ which is satisfied, if and only if the bits $a_1, \ldots, a_m$ encode the number $h < 2^m$.
	
\begin{itemize}
	\item The formula \emph{Init}:
	\begin{align*}
			\emph{Init} = \nexts{m} (&(\Delta_1(\iota_1,\ldots,\iota_{m}) \rightarrow q_\iota)\\
			& \wedge (\bigwedge \limits_{0 \le j < n} (\Delta_{j+1}(\iota_1,\ldots,\iota_{m}) \rightarrow w_j))\\
				& \wedge ( (\bigwedge \limits_{0 < j \le n}( \neg \Delta_j(\iota_1,\ldots,\iota_{m})) \rightarrow \text{\textvisiblespace})))
	\end{align*}
Note that we encode the input word $w=w_0 \cdots w_{n-1}$ in the leaves with id's $1$ to $n-1$. In our encoding the outermost leaves of the inner-trees will not encode any tape content of the Turing machine. This is due to the fact that in some inner-trees these leaves have no back-edges from which we can directly reach the leaves of the next inner-tree in DFS order. Thus, the content in the next configuration is not accessible.
 
	\item The formula \emph{Accept} checks the  last inner-tree in DFS order for an accepting configuration. The first release formula selects the outermost right path and stops at the root of the last inner-tree. If we arrive at this root we assert that the state in the configuration of this inner-tree must be accepting:

	\setcounter{equation}{0}
	\begin{align}
		&(\texttt{right} \wedge \neg \Delta_{2^h-1}(l_1,\ldots,l_{h})) {\mathcal V} ( \Delta_{2^h-1}(l_1,\ldots,l_{h})\\ 
		&\rightarrow (\texttt{right} \wedge \neg \Delta_{2^h-1}(l_1,\ldots,l_{h})) {\cal V} \\
		&( \Delta_{2^h-1}(l_1,\ldots,l_{h})  \wedge 
					\nexts{m} ( \bigvee \limits_{q \in Q} q \rightarrow \bigvee \limits_{q \in Q_F} q)))
	\end{align}
	\item We define the formula \emph{Next} to determine the successor inner-tree of an inner-tree, i.e., the formula holds at a vertex on a path if after $m$ steps on this path the root of the next inner-tree in DFS order is reached:
	\setcounter{equation}{0}
	\begin{align}
		\emph{Next}= & [(\bigvee \limits_{0<j \le m}\neg l_j ) \\
		&\rightarrow \texttt{right} \wedge (\bigwedge \limits_{0<j \le m} \nexts{j}\texttt{left})] \\
		& \wedge [(\bigwedge \limits_{0<j \le m} l_j ) \\
		&\rightarrow \texttt{left} \wedge (\bigwedge \limits_{0<j \le m} \nexts{j}\texttt{right})]
	\end{align}
	We distinguish two cases asserted by Lines (1) and (3), namely the case of an inner-tree in the maximal  level and one in a non-maximal level. In the second case, the successor tree is the left child of the current inner-tree. Here, we go up to the root of the inner-tree and traverse down the left side to the root of the left child (Line~(2)). If we are in the maximal level the successor tree is reached via the DFS order, which means, going to the inner-tree in DFS order and traversing down the right side to the right child (Line~(4)). 
	
\item If we are at a leaf with symbol $\alpha \in \Sigma$ and state $q \in Q\setminus Q_F$ we move to the root of the next inner-tree in DFS order (this is determined by the formula \emph{Next}). In this tree we have to check whether $\alpha$ is rewritten to the correct symbol and whether the head moved to the correct position. This is checked in the same way as in the proof of the unary case, namely, in three phases: we consider a path that leads to a leaf in successor inner-tree with the same leaf id, loops back up to the root of the inner-tree, and leads down to the same tree and visits the leaf to the right, loops back up and down the same tree and visits the leaf on the left (Lines (2),(3)). In Line~(4), the distance $2m+1$ results from going one step in the inner-tree and then going $2 m$ steps down to the leaves of the successor tree. The distance $3m+2$ results from looping in the same tree a second time, and $4m+3$ from looping a third time.

Here, we again use the formulas \emph{Inc} and \emph{Dec} introduced in the previous section of the appendix.
	\setcounter{equation}{0}
	\begin{align}
		\emph{config}_{q,\alpha} = \LTLsquare (& q  \wedge \alpha \wedge\emph{Next}  \wedge~~
				            (\bigwedge \limits_{j=1}^{m} \iota_j \leftrightarrow \nexts{4m+1} \iota_j)   \\ 
						  & \wedge \emph{Inc}(\iota_1,\ldots,\iota_{m}, 3m+2)) \wedge  \nexts{2m+1}(  \texttt{right}\wedge \next \texttt{root})\\
						  & \wedge \emph{Dec}(\iota_1,\ldots,\iota_{m}, 4m+3) \wedge  \nexts{3m+2}(  \texttt{right}\wedge \next \texttt{root})\\
						  & \rightarrow  ( \bigvee \limits_{_{(q',\beta, \texttt{ dir}) \in \delta(q,\alpha)}} \nexts{2m+1} \beta  \wedge \nexts{c_\texttt{dir}(m+1)-1} q' )
	\end{align}
	where $c =3$ for $\texttt{dir}=1$ and $c=4$ for $\texttt{dir}=-1$. 
	\item The formula $\emph{config}_\alpha$ is similar :
	\begin{align*}
		\emph{config}_{\alpha} = \LTLsquare (& (\bigwedge \limits_{q \in Q \setminus Q_F} \neg q )  \wedge \alpha \wedge \emph{Next} \wedge (\bigwedge \limits_{j=1}^{m} \iota_j \leftrightarrow \nexts{2m+1} \iota_j)  \rightarrow \nexts{2m+1} \alpha)
	\end{align*}
	Note that due to the DFS structure, we do not need to check for the configuration having a non-maximal id as in the unary case.
	
	\item Finally the formula $\emph{Repeat}$ propagates all final states~$q$ to the successor tree:
	\begin{align*}
	\emph{Repeat}= \LTLsquare (&(\bigvee \limits_{j=1}^{h} \neg r_j)  \wedge \emph{Next} 
	\wedge (\bigwedge \limits_{j=1}^{m} \iota_j \leftrightarrow \nexts{2m+1} \iota_j)\\
	& \rightarrow  (\bigvee \limits_{q \in Q_F}( q \rightarrow \nexts{2m+1} q))
	\end{align*}
	Again we need an additional formula that for example forces the atomic propositions to appear only in the designated node or to have only one state label in each inner-tree, and some more for other technical properties.  \qed
\end{itemize}

\end{document}